%% file: root.tex
\documentclass[conference]{IEEEtran}
\usepackage{amsmath,amssymb,amsfonts}
\def\BibTeX{{\rm B\kern-.05em{\sc i\kern-.025em b}\kern-.08em
		T\kern-.1667em\lower.7ex\hbox{E}\kern-.125emX}}

\usepackage{amsthm}
\newtheorem{theorem}{Theorem}

\newtheorem{lemma}[theorem]{Lemma}

\newtheorem{definition}{Definition}

\usepackage[english]{babel}
\usepackage{newtxtext, newtxmath}
\usepackage{cite}
\usepackage{graphicx}
\graphicspath{{Figures/}}
\usepackage[caption=false,font=footnotesize]{subfig}
\usepackage{cases}
\usepackage{mathtools}
\usepackage{xcolor}
\usepackage{url}
\usepackage{siunitx}
\usepackage{derivative}
\usepackage{booktabs}
\usepackage{bbm}

\DeclareMathOperator{\sign}{sign}

\DeclareMathOperator{\bigO}{\mathcal{O}}

\DeclarePairedDelimiter\norm{\lVert}{\rVert}
\DeclarePairedDelimiterX{\inner}[2]{\langle}{\rangle}{#1, #2}

\newcommand{\Rset}{\mathbb{R}}
\newcommand{\eps}{\varepsilon}

\usepackage{ifthen}
\newcommand{\Oinf}[1]{
\ifthenelse{
	\equal{#1}{0}}
	{\mathcal{O}(1)}
	{\mathcal{O}(\varepsilon^{#1})}
}

\newcommand{\foo}[1]{%
  \ifthenelse{\equal{\detokenize{#1}}{\detokenize{german}}}
    {TRUE}
    {FALSE}%
}

\title{Error analysis of a demodulation procedure for multicarrier signals with slowly-varying carriers}
\author{\IEEEauthorblockN{Dilshad Surroop\IEEEauthorrefmark{1}\IEEEauthorrefmark{2}, 
		Pascal Combes\IEEEauthorrefmark{2}, 
		Philippe Martin\IEEEauthorrefmark{1}}
	\IEEEauthorblockA{\IEEEauthorrefmark{1}Centre Automatique et Systèmes, MINES ParisTech, PSL University, Paris, France\\
		Email: {\tt\small\{dilshad.surroop,philippe.martin\}@mines-paristech.fr}}
	\IEEEauthorblockA{\IEEEauthorrefmark{2}Industrial Automation Business, Schneider Electric, Pacy-sur-Eure, France\\
		Email: {\tt\small pascal.combes@se.com}}}

\begin{document}
\maketitle

\begin{abstract}
We propose a procedure to demodulate analog signals encoded by a multicarrier modulator, with slowly-varying carrier shapes. We prove that the asymptotic demodulation error can be made arbitrarily small. The intended application is the ``sensorless'' control of AC electric motors at or near standstill, through the decoding of the PWM-induced current ripple.
\end{abstract}

\input{introduction.tex}

\input{definitions.tex}

\input{demodulation.tex}

\input{simulation.tex}

\input{conclusion.tex}

\bibliographystyle{phmIEEEtran}
\bibliography{biblio.bib}

%
%

\end{document}

%% file: introduction.tex
\section{Introduction}

We consider a composite signal~$y$ of the form
\begin{IEEEeqnarray}{rCl}\label{eq:MCsignal}
	y(t) &:=& \sum_{i=1}^n z_i(t)s_i\bigl(t,\tfrac{t}{\eps}\bigr)+ d\bigl(t,\tfrac{t}{\eps}\bigr)+\bigO(\eps^k),
\end{IEEEeqnarray}
where the $s_i$'s are (known) 1-periodic functions in the second variable; $\eps$ being a (known) ``small'' parameter, the $s_i$'s can been seen as rapidly-oscillating carriers with slowly-varying shapes modulating the (unknown)~$z_i$'s.
The function $d$ is a disturbance, also 1-periodic in the second variable, about which little is known except that for each~$t$ the support of $d(t,\cdot)$ is contained in a ``well-behaved'' known subset $D_t$ of~$[0,1)$. In other words, on each period of the carriers, part of the signal~$y$ is garbled and considered useless.
Finally, the $\bigO(\eps^k)$ term corresponds to ``small'' disturbances, where $\bigO$ denotes the (uniform) ``big O'' symbol of analysis, i.e. $f(t,\varepsilon)=\bigO(\eps^k)$ if $\norm{f(t,\eps)}\leq K\eps^k$ for some $K$ independent of $t$ and~$\eps$.

The objective is to recover by an implementable causal process the unknown $z_i$'s with an accuracy of up to $\bigO(\eps^k)$ from the known $y$ and $s_i$'s, provide the $s_i$'s and $z_i$'s satisfy some suitable regularity assumptions.

The motivation for this problem is the following. When operating an AC electric motor through a PWM inverter with period~$\eps$, an analysis based on the theory of averaging reveals that the currents in the motor have the form 
\begin{IEEEeqnarray*}{rCl}
	y(t) &=& y_a(t)
	+\eps y_v(t)s\bigl(t,\tfrac{t}{\eps}\bigr) 
	+\bigO(\eps^2)
	+d\bigl(t,\tfrac{t}{\eps}\bigr),
\end{IEEEeqnarray*}
which is a particular instance of~\eqref{eq:MCsignal} with $z_1:=y_a$, $z_2:=\eps y_v$, $s_1:=1$, and $s_2:=s$, where $s$ is determined by the PWM process~\cite{SurroCMR2020ACC,SurroCMR2020IECON}. The $\bigO(\eps^2)$ term corresponds to a small higher-order ripple which can be ignored. The disturbance~$d$ consists of short spikes appearing at each PWM commutation, due to stray capacitances in the power electronics. A typical (synthetic) signal~$y$ is shown in Fig.~\ref{fig:y}, see also~\cite[Fig.~9]{SurroCMR2020IECON} for experimental data. In ``sensorless'' industrial drives, these currents are the only measurements, and controlling the motor at or near standstill with this sole information is a difficult problem for several theoretical and technological reasons. A way to achieve this it to extract $y_a$ and $\eps y_v$ from the modulated currents~$y$; a suitable processing of~$y_v$ then gives access to the motor angular position~\cite{SurroCMR2020IECON}, which is instrumental in controlling the motor.
It is therefore very important to ensure the demodulation error is at most~$\bigO(\eps^2)$.

The demodulation procedure proposed in this paper, essentially consisting of multiplications by known signals followed by low-pass filters, is reminiscent of various schemes in communication theory and signal processing. Nevertheless, nothing really close seems to exist in the literature, let alone a quantitative analysis of the demodulation error:
\begin{itemize}
	\item it is of course a generalization of coherent demodulation in quadrature carrier multiplexing, with more than two carriers not restricted to sine and cosine, see e.g. \cite[section~4.4]{LathiD2010book}; but even in this simple case, no analysis of the error is usually performed, the challenges being more on carrier reconstruction
	%
	\item it somewhat looks like synchronous decorrelating detection in 
	Code-Division Multiple Access communication systems, where the $s_i$'s would play the roles of the signature waveforms and the $s_i$'s the role of the symbols, see e.g. \cite[section 5.1]{Verdu1998book}; but the encoded signals being there digital, the issues and analysis are very different
	\item it is also akin to multicarrier reception, with or without multiple access, see e.g \cite[section 12.2]{Golds2005book} 
	and~\cite[section 2.2]{Yang2009book}; 
	but once again that field is exclusively concerned with digital encoded signals
	\item finally, it bears some resemblance for its filtering part with the interpolation/compensation filters used in $\Delta\Sigma$ analog-to-digital converters, see e.g.~\cite[chapter~14]{PavanST2017book}. 
\end{itemize}

The paper extends the previous work~\cite{SurroCMR2019IECON} in two ways that are paramount for the intended application: on the one hand, it considers carriers with slowly-varying shapes, which makes the error analysis much more difficult; on the other hand, the procedure is not restricted to ``orthogonal'' demodulation, hence can directly handle the disturbance~$d$ without ad-hoc prefiltering as in~\cite{SurroCMR2020IECON}.

The  paper runs as follows: in section~\ref{sec:definitions}, we collect notations and definitions, in particular the $\mathcal{A}_k$ regularity property; in section~\ref{sec:demodulation} we state and prove the main result; in section~\ref{sec:numerical} we illustrate this result and confirm the error estimates with numerical experiments.

%% file: definitions.tex
\section{Notations and definitions}\label{sec:definitions}
We collect here definitions used throughout the paper. The most important notion is the $\mathcal{A}_k$ regularity property introduced in proposition~\ref{def:Ak}, which is needed in lemma~\ref{th:convol_Kk} to repeatedly integrate by parts; this property, which is paramount for handling carriers with slowly-varying shapes, is trivially satisfied for fixed-shape carriers as in~\cite{SurroCMR2019IECON}.

Let $g(t,\sigma)$ be a function of two variables; informally speaking, $t$ represents the slow timescale and $\sigma$ the fast timescale. We will often use the convenient notation $g_\varepsilon(t):=g\bigl(t,\frac{t}{\eps}\bigr)$. 

The function $g$ is $1$-periodic in the second variable if $g(t,\sigma+1)=g(t,\sigma)$ for all~$t$. Its mean in the second variable is the function $\overline{g}(t):=\int_0^1g(t,\sigma)d\sigma$. For brevity, we will usually omit the phrase ``in the second variable''.
If $g$ is 1-periodic with zero mean, any of its primitives (in the second variable) is also 1-periodic, in particular its zero-mean primitive $g^{(-1)}(t,\sigma) := \int_0^\sigma g(t,\tau)d\tau-\int_0^1 \int_0^\zeta g(t,\tau)d\tau d\zeta$. Likewise, $g^{(-k-1)}$ denotes the zero-mean primitive of $g^{(-k)}$.

We say $g$ is Lipschitz if $\norm{g(t_1,\sigma)-g(t_2,\sigma)}\le L\norm{t_1-t_2}$ for some $L$ independent of $t_1$, $t_2$ and~$\sigma$.

Finally, we introduce the $\mathcal{A}_k$ regularity property.
\begin{definition}[$\mathcal{A}_k$ property]\label{def:Ak}
	Let $g(t,\sigma)$ be 1-periodic with zero mean. It is said to be $\mathcal{A}_k$, $k\ge1$, if $g^{(-k)}$ is $k-1$ times differentiable in the first variable, with bounded derivatives at all orders,
and $\partial_1^{k-1}g^{(-k)}$ 
Lipschitz.
\end{definition}
A typical $\mathcal{A}_k$ function encountered in practice is $g(t,\sigma)=\sign\bigl(u(t)-\boldsymbol{\sigma}\bigr)-2u(t)+1$
where $\boldsymbol{\sigma}:=\sigma\mod1$; $u(t)\in(0,1)$ represents the PWM duty cycle and is assumed $k-1$ times differentiable, with bounded derivatives at all orders, and $u^{(k-1)}$ Lipschitz.

It is easy to show that if on the one hand $g(t,\sigma)$ is~$\mathcal{A}_k$, and on the other hand $z(t)$ is $k-1$ times differentiable, with bounded derivatives at all orders, and $z^{(k-1)}$ Lipschitz, then the product $zg$ is also~$\mathcal{A}_k$.

%% file: demodulation.tex
\section{The demodulation procedure}\label{sec:demodulation}
The demodulation procedure for an error of order~$\eps^k$ consists of multiplications by a suitable demodulating basis $R:=(r_1, \ldots, r_n)^T$, followed by a bank of low-pass finite impulse response filters with kernel~$\widetilde K_k$; see section~\ref{sec:main} for a discussion
of how to select~$R$.
The kernel $\widetilde K_k$ is a ``compensated'' $k$-times iterated moving average, namely a suitable linear combination of shifted instances of~$K_k$, where the kernel $K_k$ is defined recursively by $K_1:= \frac{1}{\eps}\mathbbm{1}_{[0,\eps]}$ and $K_k:=K_{k-1}*K$, see e.g.~\cite[chapter~$6.7$]{Aubin2000book} for explicit expressions. For instance for $k=3$, the linear combination is
\begin{IEEEeqnarray*}{rCl}
	\widetilde K_3(t) &:=& \tfrac{17}{4}K_3(t) -5K_3(t-\eps) + \tfrac{7}{4}K_3(t-2\eps),
\end{IEEEeqnarray*}
see section~\ref{sec:proof_thm} for more details.

\begin{figure}[ht]
	\centering
	\includegraphics[scale=1]{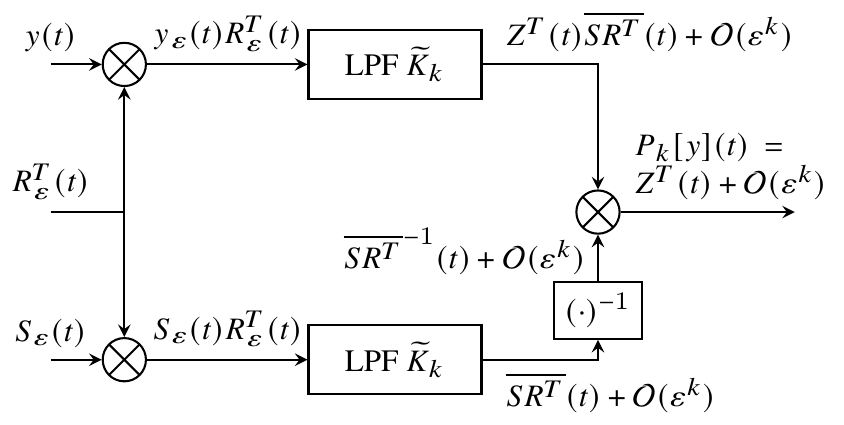}
	\caption{The demodulation procedure.}
	\label{fig:diagram}
\end{figure}

Fig.~\ref{fig:diagram} illustrates the whole demodulation procedure:
\begin{itemize}
	\item $y(t)$ is multiplied by $R_\eps^T(t)$, and filtered by~$\widetilde{K}_k$; 
	the result, $\Bigl(\widetilde K_k*\bigl(yR^T_\eps\bigr)\Bigr)(t)$, turns out to be $Z^T(t)\overline{SR^T}(t)+\bigO(\eps^k)$, where $Z:=(z_1,\ldots,z_n)^T$ is the vector signal to recover
	\item the modulating basis $S:=(s_1,\ldots,s_n)^T$ is also multiplied by $R_\eps^T(t)$, and filtered by~$\widetilde{K}_k$; 
	the result, $\Bigl(\widetilde K_k*\bigl(S_\eps R_\eps^T\bigr)\Bigr)(t)$, turns out to be $\overline{SR^T}(t) +\bigO(\eps^k)$
	\item finally, $\Bigl(\widetilde K_k*\bigl(yR^T_\eps\bigr)\Bigr)(t)$ is multiplied by the inverse of the matrix $\Bigl(\widetilde K_k*\bigl(S_\eps R_\eps^T\bigr)\Bigr)(t)$; the result, $\Bigl(\widetilde K_k*\bigl(y R_\varepsilon^T\bigr)\Bigr)(t)
	\times \Bigl(\widetilde K_k*\bigl(S_\eps R_\eps^T\bigr)\Bigr)^{-1}(t)$, is as desired $Z^T(t)+\bigO(\eps^k)$.
\end{itemize}
As pointed out in the introduction, this demodulation scheme is at first sight not completely surprising. What is much less obvious is that the overall demodulation error is indeed of order~$\eps^k$.

\subsection{Main result}\label{sec:main}
We assume that the $s_i$'s are independent outside the subset~$D_t$ containing the support of the disturbance~$d(t,\cdot)$,
i.e. that the $\check s_i$'s defined by
$\check s_i(t,\sigma):=\bigl(1-\mathbbm{1}_{D_t}\bigr)(\sigma)s_i(t,\sigma)$
are linearly independent. We can thus choose the demodulating basis $R:=(r_1, \ldots, r_n)^T$ such that $Rd=0$ and $\overline{SR^T}$ is invertible,
where $S:=(s_1,\ldots,s_n)^T$ is the modulating basis; one simple choice is for instance $R(t,\sigma):=\bigl(1-\mathbbm{1}_{D_t}(\sigma)\bigr)S(t,\sigma)$.
A delicate point is to select~$R$ also such that $SR^T-\overline{SR^T}$ is $\mathcal{A}_k$, provided of course that $D_t$ is ``well-behaved'' (for instance a finite union of intervals with sufficiently regular moving bounds). For simplicity, we just assume this is the case (and check it a posteriori in the numerical experiments of section~\ref{sec:numerical}).
Finally, we assume the $z_i$'s are $k-1$ times differentiable, with bounded derivatives at all orders, and $z_i^{(k-1)}$ Lipschitz,
so that $z_i\bigl(SR^T-\overline{SR^T}\bigr)$ is also~$\mathcal{A}_k$.
\begin{theorem}\label{thm:main}
$Z:=(z_1,\ldots,z_n)^T$ can be recovered to order~$\eps^k$ from~$y$ by the causal process $P_k$ defined by
\begin{IEEEeqnarray*}{rCl}
	P_k[y](t) &:=& \Bigl(\widetilde K_k*\bigl(y R_\varepsilon^T\bigr)\Bigr)(t)
	\times \Bigl(\widetilde K_k*\bigl(S_\eps R_\eps^T\bigr)\Bigr)^{-1}(t).
\end{IEEEeqnarray*}
In other words, $Z^T(t)=P_k[y](t)+\bigO(\eps^k)$.
\end{theorem}
\input{proof.tex}
\input{lemmas.tex}

%% file: proof.tex
\subsection{Proof of theorem~\ref{thm:main}}\label{sec:proof_thm}
\label{subsec:reconstruction}

Rewriting~\eqref{eq:MCsignal} as
\begin{IEEEeqnarray*}{rCl}
	y(t) &=& Z^T(t)S\bigl(t,\tfrac{t}{\eps}\bigr)+ d\bigl(t,\tfrac{t}{\eps}\bigr)+\bigO(\eps^k),           
\end{IEEEeqnarray*}
right-multiplying by $R^T_\eps$ and convolving with~$K_k$ yields
\begin{IEEEeqnarray*}{rCl}
	\Bigl(K_k*\bigl(yR^T_\eps\bigr)\Bigr)(t) &=& \Bigl(K_k*\bigl(Z^TS_\eps R^T_\eps\bigr)\Bigr)(t)+\bigO(\eps^k)\\
	&=& \Bigl[K_k*\Bigl(Z^T\bigl(S_\eps R^T_\eps-\overline{SR^T}\bigr)\Bigr)\Bigr](t)\\
	&&\quad+\,\Bigl(K_k*\bigl(Z^T\overline{SR^T}\bigr)\Bigr)(t)+\bigO(\eps^k)\\
	&=& \Bigl(K_k*\bigl(Z^T\overline{SR^T}\bigr)\Bigr)(t)+\bigO(\eps^k);
\end{IEEEeqnarray*}
to obtain the last line, we have applied Lemma~\ref{th:convol_Kk}
with $g(t,\sigma):=Z^T(t)\bigl(S(t,\sigma)R^T(t,\sigma)-\overline{SR^T}(t)\bigr)$,
which is by construction zero-mean and~$\mathcal{A}_k$.
The result obviously holds also if $K_k(t)$ is replaced by the shifted kernel $\tau_TK_k(t):=K_k(t-T)$.

On the other hand, \cite[Theorem~1]{SurroCMR2019IECON} asserts that a $\mathcal{C}^k$-function $\varphi$ with bounded $\varphi^{(k)}$
is left unchanged to order~$\eps^k$ by a suitable linear combination $\widetilde K_k$ of the shifted kernels $\tau_{i\eps}K_k$, $i=0,\ldots,k-1$, i.e.
$(\widetilde K_k*\varphi)(t)=\varphi(t)+\bigO(\eps^k)$.
For instance,
\begin{IEEEeqnarray*}{rCl}
	\widetilde K_1(t) &:=& K_1(t)\\
	\widetilde K_2(t) &:=& 2K_2(t)-K_2(t-\eps), \\
	\widetilde K_3(t) &:=& \tfrac{17}{4}K_3(t) -5K_3(t-\eps) + \tfrac{7}{4}K_3(t-2\eps).
\end{IEEEeqnarray*}
Actually, we must slightly extend the result to the case where $\varphi$ is $k-1$ times differentiable with $\varphi^{(k-1)}$ Lipschitz, which we omit by lack of space.
As a consequence,
\begin{IEEEeqnarray*}{rCl}
	\Bigl(\widetilde K_k*\bigl(yR^T_\eps\bigr)\Bigr)(t) 
	&=& \Bigl(\widetilde K_k*\bigl(Z^T\overline{SR^T}\bigr)\Bigr)(t)+\bigO(\eps^k)\\
	&=& Z^T(t)\overline{SR^T}(t)+\bigO(\eps^k).
\end{IEEEeqnarray*}
Since $\overline{SR^T}(t)$ is invertible, $Z(t)$ can be recovered to order~$\eps^k$.

To make the process truly implementable in practice, notice $\overline{SR^T}(t)$ can be computed to order~$\eps^k$ by
\begin{IEEEeqnarray*}{rCl}
	\Bigl(\widetilde K_k*\bigl(S_\eps R_\eps^T\bigr)\Bigr)(t) &=& \overline{SR^T}(t) +\bigO(\eps^k), 
\end{IEEEeqnarray*}
which is an instance of the previous equation with $Z$ constant.

In conclusion, $Z(t)$ is recovered to order~$\eps^k$ by
\begin{IEEEeqnarray*}{rCl}
	P_k[y](t) &:=& \Bigl(\widetilde K_k*\bigl(y R_\varepsilon^T\bigr)\Bigr)(t)
	\times \Bigl(\widetilde K_k*\bigl(S_\eps R_\eps^T\bigr)\Bigr)^{-1}(t)\\
	&=& Z^T(t) +\bigO(\eps^k),
\end{IEEEeqnarray*}
where the process~$P_k$ is causal since the kernel $\widetilde K_k$ is supported on~$[0,k\eps]\subset\Rset^+$.

%% file: lemmas.tex
\subsection{Technical lemmas}
This section is quite technical and can be skipped without disturbing the flow of ideas. Its goal is to establish Lemma~\ref{th:convol_Kk}, which is instrumental in the proof of Theorem~\ref{thm:main}. Lemma~\ref{th:convol_Kk} relies on Lemma~\ref{lemma:delta}, which itself relies on Lemma~\ref{lemma:deltak}. Lemmas \ref{th:convol_Kk} and~\ref{lemma:delta} are in some sense properties of the convolution kernel~$K_k$, whereas Lemma~\ref{lemma:deltak} extends to our context a classical result of finite-differences calculus. Notice the use of the $\mathcal{A}_k$ property when integrating by parts in Lemma~\ref{th:convol_Kk}, which is the main trick to extend the ideas of~\cite{SurroCMR2019IECON} to slowly-moving carriers.

Define the $k^\text{th}$-order backward difference $\Delta_k g_\eps$ of the function $g_\eps(t):=g(t,\frac{t}{\eps})$ by 
\begin{IEEEeqnarray*}{rCl}
	(\Delta_k g_\eps)(t) &:=& \sum_{i=0}^k (-1)^i \binom{k}{i} g_\eps(t-i\eps).
\end{IEEEeqnarray*}
On the other hand, recall that $K_k$ is $k-1$ times differentiable, with compact support for all the derivatives. As for $K_k^{(k)}$, it can be defined in the distributional sense, and is a linear combination of Dirac delta functions, and in particular also has compact support; for instance, $K_1^{(1)}=\frac{1}{\eps}\bigl(\delta_0-\delta_\eps\bigr)$.
\begin{lemma}
	\label{lemma:deltak}
	Let $g(t,\sigma)$ be 1-periodic, and $k-1$ times differentiable in the first variable with $\partial_1^{(k-1)}g$ Lipschitz. Then $(\Delta_k g_\eps)(t) = \bigO(\eps^k)$.
\end{lemma}
\begin{proof}
	By the Lipschitz form of Taylor's formula~\cite[(2.1)]{EllisSD1990JDE},
	\begin{IEEEeqnarray*}{rCl}
		g(t+\mu,\sigma) &=& \sum_{j=0}^{k-1} \frac{\mu^j}{j!}\partial_1^j g(t,\sigma)  + \mu^k \rho_t(\mu,\sigma),
	\end{IEEEeqnarray*}
	where the remainder $\rho_t$ is $\bigO(1)$ since it satisfies
	\begin{IEEEeqnarray*}{rCl}
		\mu\rho_t(\mu,\sigma) &=& \frac{1}{(k-2)!} \int_0^1 (1-\tau)^{(k-2)}\\ &&\quad\times\,\bigl(\partial_1^{k-1}g(t+\mu\tau,\sigma) -\partial_1^{k-1} g(t,\sigma)\bigr)d\tau.
	\end{IEEEeqnarray*}
	Applying this to $g\bigl(t-i\eps, \tfrac{t-i\eps}{\eps} \bigr)=g\bigl(t-i\eps, \tfrac{t}{\eps} \bigr)$ since $g$ is 1-periodic yields
	\begin{IEEEeqnarray*}{rCl}
		(\Delta_k g_\eps)(t) 
		&=& \sum_{i=0}^k (-1)^i \binom{k}{i} g\bigl(t-i\eps, \tfrac{t-i\eps}{\eps} \bigr)\\
		&=& \sum_{i=0}^k (-1)^i \binom{k}{i} \Biggl(\sum_{j=0}^{k-1} 	\frac{(-i\eps)^j}{j!}\partial_1^j g(t,\tfrac{t}{\eps})+ \bigO(\eps^k)\Biggr)\\
		&=&	\sum_{j=0}^{k-1}\frac{(-\eps)^j}{j!}\partial_1^j g(t,\tfrac{t}{\eps})
		\sum_{i=0}^k (-1)^i\binom{k}{i}i^j+\bigO(\eps^k)
	\end{IEEEeqnarray*}
	As $\sum_{i=0}^k (-1)^i\binom{k}{i}i^j=0$, see \cite[Cor.~2]{Ruiz1996TMG}, this gives the desired result.
\end{proof}
%
%
\begin{lemma}
	\label{lemma:delta}
	Let $g(t,\sigma)$ be 1-periodic, and $k-1$ times differentiable in the first variable with $\partial_1^{(k-1)}g$ Lipschitz. Then $\bigl(K_k^{(k)}*g_\eps\bigr)(t)=\bigO(1)$.
\end{lemma}
\begin{proof}
	We first prove by induction that $K_k^{(k)}*g_\eps=\frac{1}{\eps^k}\Delta_kg_\eps$. Indeed, $K_1'*g_\eps=\frac{1}{\eps}\bigl(\delta_0-\delta_\eps\bigr)*g_\eps=\frac{1}{\eps}\Delta_1g_\eps$. Assuming the property holds at rank~$k$,
	\begin{IEEEeqnarray*}{rCl}
		K_{k+1}^{(k+1)}*g_\eps &=& (K_k*K_1)^{(k+1)}*g_\eps\\
		&=& K_k^{(k)}*K_1'*g_\eps\\
		&=& \frac{1}{\eps^k}\Delta_k(K_1'*g_\eps)\\
		&=& \frac{1}{\eps^k}\Delta_k\Bigl(\frac{\Delta_1g_\eps}{\eps}\Bigr)\\
		&=& \frac{1}{\eps^{k+1}}\Delta_{k+1}g_\eps. 
	\end{IEEEeqnarray*}
	To obtain the second line, we have repeatedly used $(T*S)'=T'*S=T*S'$.
	
	Applying Lemma~\ref{lemma:deltak}, we eventually find $(K_k^{(k)}*g_\eps)(t)=\frac{1}{\eps^k}\Delta_kg_\eps(t)=\frac{1}{\eps^k}\bigO(\eps^k)=\bigO(1)$.
\end{proof}

\begin{lemma}
	\label{th:convol_Kk}
	Let $g$ be~$\mathcal{A}_{k}$. Then $\bigl(K_k*g_\eps\bigr)(t) = \bigO(\eps^k)$
\end{lemma}
\begin{proof}
	We sketch the proof for $k=3$, the general result following by induction.
	Notice that by assumption $g^{(-3)}$ is twice differentiable in the first variable with $\partial_1^2g^{(-3)}$ Lipschitz, which will be used each time Lemma~\ref{lemma:delta} is invoked.
	
	We first prove $\Bigl(K_3''*\bigl(g^{(-2)}\bigr)_\eps\Bigr)(t)=\bigO(\eps)$. Starting from
	\begin{IEEEeqnarray*}{rCl}
		\Bigl(\bigl(g^{(-3)}\bigr)_\eps\Bigr)' 
		&=& \bigl(\partial_1g^{(-3)}\bigr)_\eps + \frac{1}{\eps}\bigl(\partial_2g^{(-3)}\bigr)_\eps\\
		&=& \bigl(\partial_1g^{(-3)}\bigr)_\eps + \frac{1}{\eps}\bigl(g^{(-2)}\bigr)_\eps,
	\end{IEEEeqnarray*}
	we find after convolving with $K_3''$ and integrating by parts
	\begin{IEEEeqnarray*}{rCl}
		K_3''*\bigl(g^{(-2)}\bigr)_\eps &=& \eps K_3'''*\bigl(g^{(-3)}\bigr)_\eps - \eps K_3''*\bigl(\partial_1g^{(-3)}\bigr)_\eps;
	\end{IEEEeqnarray*}
	the boundary terms vanish since $K_3''$ has compact support.
	The first term is $\bigO(\eps)$ by Lemma~\ref{lemma:delta}. Using $K_3''=(K_1*K_2)''=K_1*K_2''$,
	the second term reads $\eps K_1*\Bigl(K_2''*\bigl(\partial_1g^{(-3)}\bigr)_\eps\Bigr)$, hence is also 
	$\bigO(\eps)$ by Lemma~\ref{lemma:delta}. The sum of the two terms is therefore also~$\bigO(\eps)$.
	
	We next prove $\Bigl(K_3'*\bigl(g^{(-1)}\bigr)_\eps\Bigr)(t)=\bigO(\eps^2)$. Indeed, using successively
	\begin{IEEEeqnarray*}{rCl}
		\bigl(g^{(-1)}\bigr)_\eps &=& \eps\bigl(g^{(-2)}\bigr)_\eps' - \eps\bigl(\partial_1g^{(-2)}\bigr)_\eps\\
		\bigl(\partial_1g^{(-2)}\bigr)_\eps &=& \eps\bigl(\partial_1g^{(-3)}\bigr)_\eps' - \eps\bigl(\partial_1^2g^{(-3)}\bigr)_\eps
	\end{IEEEeqnarray*}
	yields
	\begin{IEEEeqnarray*}{rCl}
		\bigl(g^{(-1)}\bigr)_\eps &=& \eps\bigl(g^{(-2)}\bigr)_\eps'
		- \eps^2\bigl(\partial_1g^{(-3)}\bigr)_\eps'
		+ \eps^2\bigl(\partial_1^2g^{(-3)}\bigr)_\eps.
	\end{IEEEeqnarray*}
	Convolving with $K_3'$ and integrating by parts,
	\begin{IEEEeqnarray*}{rCl}
		K_3'*\bigl(g^{(-1)}\bigr)_\eps &=& \eps K_3''*\bigl(g^{(-2)}\bigr)_\eps
		- \eps^2K_3''*\bigl(\partial_1g^{(-3)}\bigr)_\eps\\
		&&\quad+\, \eps^2K_3'*\bigl(\partial_1^2g^{(-3)}\bigr)_\eps;
	\end{IEEEeqnarray*}
	the boundary terms vanish since $K_3'$ has a bounded support.
	We already know the first two terms are~$\bigO(\eps^2)$. Using $K_3'=(K_2*K_1)'=K_2*K_1'$,
	the last term reads $\eps^2K_2*\Bigl(K_1'*\bigl(\partial_1^2g^{(-3)}\bigr)_\eps\Bigr)$, hence is also 
	$\bigO(\eps^2)$ by lemma~\ref{lemma:delta}. The sum of the three terms is therefore also~$\bigO(\eps^2)$.
	
	We finally prove $\bigl(K_3*g_\eps\bigr)(t)=\bigO(\eps^3)$. Indeed, using successively
	\begin{IEEEeqnarray*}{rCl}
		g_\eps &=& \eps\bigl(g^{(-1)}\bigr)_\eps' - \eps\bigl(\partial_1g^{(-1)}\bigr)_\eps\\
		\bigl(\partial_1g^{(-1)}\bigr)_\eps &=& \eps\bigl(\partial_1g^{(-2)}\bigr)_\eps' - \eps\bigl(\partial_1^2g^{(-2)}\bigr)_\eps\\
		\bigl(\partial_1^2g^{(-2)}\bigr)_\eps &=& \eps\bigl(\partial_1^2g^{(-3)}\bigr)_\eps' - \eps\bigl(\partial_1^3g^{(-3)}\bigr)_\eps,
	\end{IEEEeqnarray*}
	we find
	\begin{IEEEeqnarray*}{rCl}
		g_\eps &=& \eps\bigl(g^{(-1)}\bigr)_\eps'
		- \eps^2\bigl(\partial_1g^{(-2)}\bigr)_\eps'\\
		&&\quad+\, \eps^3\bigl(\partial_1^2g^{(-3)}\bigr)_\eps'
		- \eps^3\bigl(\partial_1^3g^{(-3)}\bigr)_\eps.
	\end{IEEEeqnarray*}
	Convolving with $K_3$ and integrating by parts,
	\begin{IEEEeqnarray*}{rCl}
		K_3*g_\eps &=& \eps K_3'*\bigl(g^{(-1)}\bigr)_\eps	- \eps^2K_3'*(\partial_1g^{(-2)}\bigr)_\eps\\
		&&\quad+\, \eps^3K_3'*\bigl(\partial_1^2g^{(-3)}\bigr)_\eps
		- \eps^3K_3*\bigl(\partial_1^3g^{(-3)}\bigr)_\eps;
	\end{IEEEeqnarray*}
	the boundary terms vanish since $K_3$ has a bounded support.
	We already know the first and third terms are~$\bigO(\eps^3)$. The second term reads
	\begin{IEEEeqnarray*}{rCl}
		\eps^2K_3'*(\partial_1g^{(-2)}\bigr)_\eps
		&=& \eps^3K_3'*\bigl(\partial_1g^{(-3)}\bigr)_\eps' - \eps^3K_3'*\bigl(\partial_1^2g^{(-3)}\bigr)_\eps\\
		&=& \eps^3K_3''*\bigl(\partial_1g^{(-3)}\bigr)_\eps - \eps^3K_3'*\bigl(\partial_1^2g^{(-3)}\bigr)_\eps\\
		&=& \eps^3K_1*\Bigl(K_2''*\bigl(\partial_1g^{(-3)}\bigr)_\eps\Bigr)\\
		&&\quad -\, \eps^3K_2'*\Bigl(K_1'*\bigl(\partial_1^2g^{(-3)}\bigr)_\eps\Bigr),
	\end{IEEEeqnarray*}
	and is also $\bigO(\eps^3)$ by using Lemma~\ref{lemma:delta} twice. Finally, by Young's convolution inequality, 
	the fourth term satisfies
	\begin{IEEEeqnarray*}{rCl}
		\norm{\eps^3K_3*\partial_1^3g^{(-3)}}_\infty \le \eps^3\norm{K_3}_1\norm{\partial_1^3g^{(-3)}}_\infty,
	\end{IEEEeqnarray*}
	hence is also $\bigO(\eps^3)$; notice $\partial_1^3g^{(-3)}$ is bounded by assumption, and so is~$K_3$.
	The sum of the four terms is therefore also~$\bigO(\eps^3)$, which concludes the proof.
\end{proof}

\vspace{-5mm}
\begin{figure}[ht]
	\centering
	\includegraphics[scale=0.95]{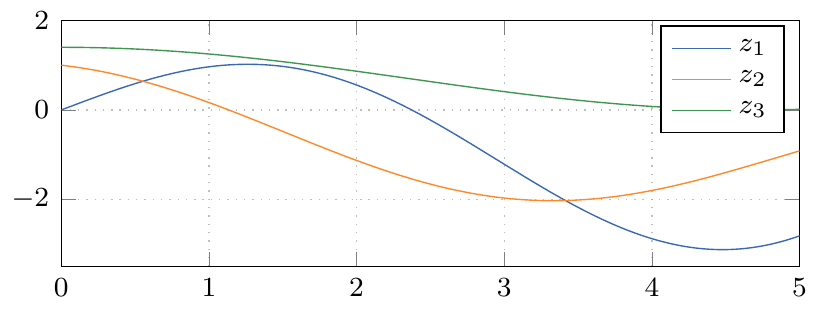}
	\caption{Encoded signals $z_1(t),z_2(t),z_3(t)$.}
	\label{fig:Y}
\end{figure}\vspace{-5mm}
\begin{figure}[ht]
	\centering
	\includegraphics[scale=0.95]{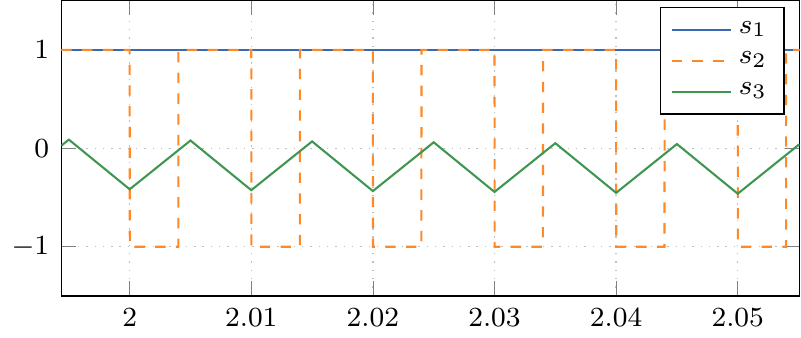}
	\caption{Carriers $s_1(t),s_2(t),s_3(t)$ (zoom).}
	\label{fig:S}
\end{figure}\vspace{-5mm}
\begin{figure}[ht]
	\centering
	\includegraphics[scale=0.95]{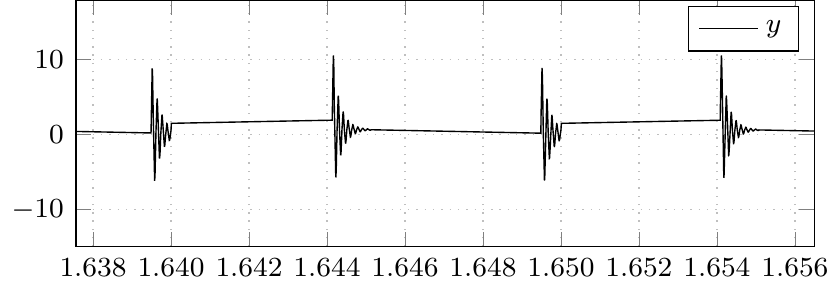}
	\caption{Composite signal~$y(t)$ (zoom).}
	\label{fig:y}
\end{figure}

%% file: simulation.tex
\section{Numerical experiments}\label{sec:numerical}
We illustrate the error analysis of Theorem~\ref{thm:main} with numerical experiments for~$k=1,2,3$.
\begin{figure}[ht]
	\centering
	\includegraphics[scale=0.95]{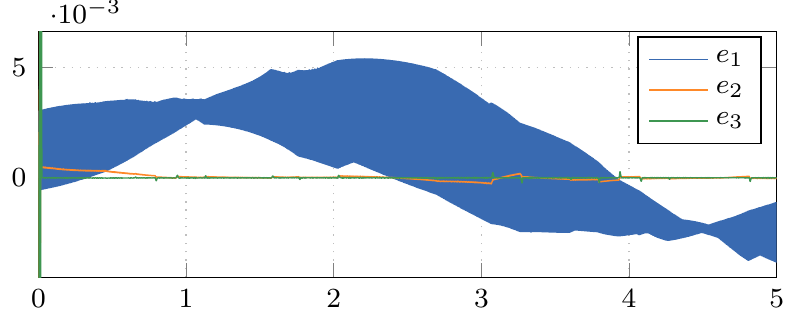}
\includegraphics[scale=0.95]{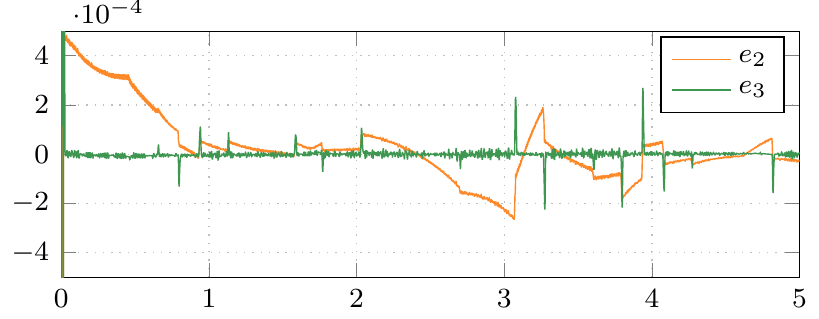}
	\caption{Errors $e_k(t)$ for $\eps = 10^{-2}$: full view (top), zoom (bottom).}
	\label{fig:e2}
\end{figure}
Consider the composite signal $y$ defined on $[0,5]$ by
\begin{align*}
	y(t) &= z_1(t)s_1(t,\tfrac{t}{\varepsilon})+ z_2(t)s_2(t,\tfrac{t}{\varepsilon})+
	z_3(t)s_3(t,\tfrac{t}{\varepsilon}) + d(t,\tfrac{t}{\varepsilon}),
\end{align*}
with encoded signals $z_1,z_2,z_3$ (see Fig.~\ref{fig:Y})
\begin{IEEEeqnarray*}{rCl}
	z_1(t) &:=& 2\sin(t) - 1.5\sin\bigl(\tfrac{t}{2}\bigr)\\
	z_2(t) &:=& \cos (t)-1.2\sin\bigl(\tfrac{t}{4}\bigr)\\
	z_3(t) &:=& 1.4\cos(\tfrac{t}{3})^2;
\end{IEEEeqnarray*}
and carriers $s_1,s_2,s_3$ (see Fig.~\ref{fig:S})
\begin{IEEEeqnarray*}{rCl}
	s_1(t,\sigma) &:=& 1\\
	s_2(t,\sigma) &:=& \sign\bigl(\tfrac{t}{20} + \boldsymbol{\sigma} -0.5\bigr)\\
	s_3(t,\sigma) &:=& \begin{cases}
		\cos (t)+ \boldsymbol{\sigma} & \boldsymbol{\sigma} \leq 0.5 \\
		\cos(t)+1 -\boldsymbol{\sigma} & \boldsymbol{\sigma} \geq 0.5, \\
	\end{cases}
\end{IEEEeqnarray*}
where $\boldsymbol{\sigma}:=\sigma\mod1$. The support of the disturbance~$d$
is
\begin{IEEEeqnarray*}{rCl}
	D_t &:=& \bigl[f(t)-\tfrac{1}{20},f(t)+\tfrac{1}{20}\bigr]
	\cup\bigl[g(t)-\tfrac{1}{20},g(t)+\tfrac{1}{20}\bigr],
\end{IEEEeqnarray*}
with $f(t):=\tfrac{1}{2}\bigl(1+\sin (t)\bigr)$ and $g(t):=\tfrac{1}{2}\bigl(1+\cos (t)\bigr)$; hence, on a window of length~$\eps$ between \SI{10}{\percent} (when the two intervals coincide) and \SI{20}{\percent} (when the two intervals are disjoint) of the signal is corrupted. Fig~\ref{fig:y} displays the resulting signal~$y$, with the spikes caused by~$d$ clearly visible.

We select the simplest demodulating basis that is zero on~$D_t$, namely 
$R(t,\sigma):=\bigl(1-\mathbbm{1}_{D_t}(\sigma)\bigr)S(t,\sigma)$; tedious but routine computations show $SR^T-\overline{SR^T}$ is $\mathcal{A}_k$ for $k=1,2,3$.
We check numerically that $\overline{SR^T}(t)$ is invertible by plotting its condition number~$\kappa$, see Fig.~\ref{fig:cond}: indeed, $\overline{SR^T}(t)$ is always well-conditioned, except during the filter initialization.
%
%
%
%

	\begin{figure}[ht]
	\centering
	\includegraphics[scale=0.95]{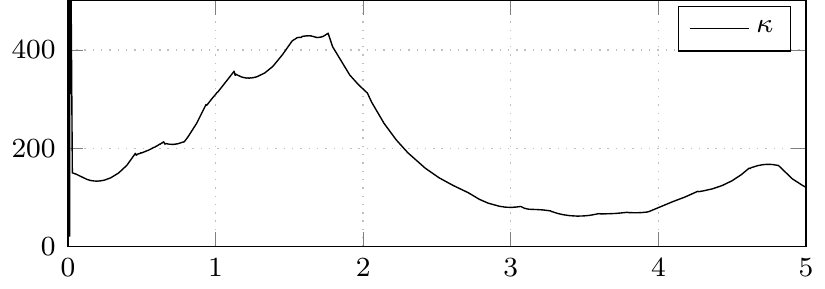}
	\caption{Condition number $\kappa(t)$ of matrix~$\overline{SR^T}(t)$.}
	\label{fig:cond}
\end{figure}

\begin{figure}[ht]
	\centering
	\includegraphics{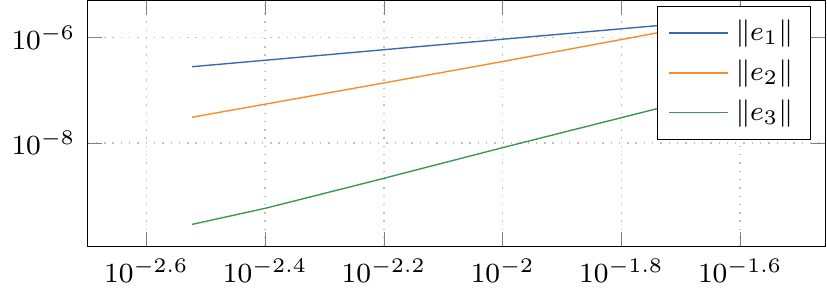}
	\caption{$L^2$-error $\norm{e_k}$ as a function of $\eps$.}
	\label{fig:asymp}
\end{figure}


We focus on the recovery of~$z_2$, since it is modulated by the least regular carrier. We consider the error $e_k(t):=z_2(t) - P_k^2[y](t)$, where $P_k^2[y]$ denotes the second component of~$P_k[y]$. For $\eps$ fixed, the error decreases as anticipated with~$k$, see Fig~\ref{fig:e2}.  To study the asymptotic behavior as a function of~$\eps$, we consider the $L_2$-error $\norm{e_k}:= \bigl(\int_1^5 \bigl(e_k(t)\bigr)^2dt\bigr)^\frac{1}{2}$; the first second of data is discarded to ensure the filters are well initialized. As anticipated, the plots in log scale are straight lines with slopes equal to the orders of the estimates, see~Fig.~\ref{fig:asymp}.

%% file: conclusion.tex
\section{Conclusion}
We have proposed a demodulation procedure to recover analog signals encoded by multiple carriers with slowly-varying shapes. Though the procedure is not completely surprising at first sight, proving that the overall demodulation error is arbitrarily small is not obvious. Arguably, the framework is somewhat peculiar, which explains why no similar work seems to exist in the literature. Nevertheless, the result is exactly what we need for the application we have in mind, namely the ``sensorless'' control of AC electric motors at or near standstill. 
In this application, the composite signal~$y$ to be decoded is the (vector) current in the motor, the motor itself acting as a multicarrier modulator when fed by a PWM inverter; a suitable processing of the demodulated signal then yields the rotor angle, which is needed to accurately control the motor.
